\documentclass{article}

\usepackage{lmodern}
\usepackage[T1]{fontenc}

\usepackage{latexsym}
\usepackage{amssymb}
\usepackage{amsthm}
\usepackage{amsmath}

\usepackage{amstext}
\usepackage{amsopn}

\usepackage{enumerate}
\usepackage{graphics}
\usepackage{epsf}

\usepackage{tikz}
\usepackage[all]{xy}
\usepackage{url}

\newtheorem{theorem}{Theorem}

\newtheorem{lemma}[theorem]{Lemma}
\newtheorem {example}[theorem]{Example}
\newtheorem{corollary}[theorem]{Corollary}
\newtheorem{note}[theorem]{Note}

\newcommand{\sgraph}{G}

\newcommand{\weight}{w}
\newcommand{\neighbour}{N}
\newcommand{\products}{\mathcal{P}}

\newcommand{\snet}{\mathcal{S}}
\newcommand{\prodset}{P}

\newcommand{\obar}[1]{\overline{#1}}

\newcommand{\srcnodes}{\mathit{source}}
\newcommand{\payoff}{p}
\newcommand{\strprofile}{s}

\newcommand{\agents}{\mathcal{A}}
\newcommand{\spred}{\operatorname{Pred}}
\newcommand{\nat}{\mathbb{N}}

\newcommand{\inflset}{\mathcal{N}}

\newcommand{\constutil}{c_0}

\newcommand{\ES}{\emptyset}

\newcounter{symbol}
\newcommand{\indexsyma}[1]%
{\stepcounter{symbol}\index{zzz1 \thesymbol @\protect#1}}
\newcommand{\indexsymb}[1]%
{\stepcounter{symbol}\index{zzz2 \thesymbol @\protect#1}}
\newcommand{\indexsymc}[1]%
{\stepcounter{symbol}\index{zzz3 \thesymbol @\protect#1}}
\newcommand{\indexsymd}[1]%
{\stepcounter{symbol}\index{zzz4 \thesymbol @\protect#1}}
\newcommand{\indexsyme}[1]%
{\stepcounter{symbol}\index{zzz5 \thesymbol @\protect#1}}

\newcommand{\bfe}[1]{\begin{bfseries}\emph{#1}\end{bfseries}\index{#1}}
\newcommand{\oldbfe}[1]{\begin{bfseries}\emph{#1}\end{bfseries}}

\newcommand{\myra}{\mbox{$\:\rightarrow\:$}}

\newcommand{\sse}{\mbox{$\:\subseteq\:$}}

\newcommand{\fa}{\mbox{$\forall$}}
\newcommand{\te}{\mbox{$\exists$}}

\newcommand{\LLn}{\mbox{$1,\ldots,n$}}
\newcommand{\LL}{\mbox{$\ldots$}}

\newcommand{\C}[1]{\mbox{$\{{#1}\}$}}           

\newcommand{\NI}{\noindent}
\newcommand{\HB}{\hfill{$\Box$}}

\newcommand{\II}{\vspace{2 mm}}

\newcommand{\szkew}[1]{\relax \setbox0=\hbox{\kern -24pt $\displaystyle#1$\kern 0pt }%
\box0}
{\catcode`\@=11 \global\let\ifjusthvtest@=\iffalse}

\newcounter{oldmycaption}

\title{Paradoxes in Social Networks with Multiple Products}
\author{Krzysztof R. Apt \\
\emph{Centre for Mathematics and Computer Science (CWI)} \\
\emph{and University of Amsterdam, The Netherlands} \\
\and
Evangelos Markakis \\
\emph{Athens University of Economics and Business, Athens, Greece} \\
\and
Sunil Simon \\
\emph{Centre for Mathematics and Computer Science (CWI)}
}

\begin{document}

\maketitle

\begin{abstract}
  Recently, we introduced in \cite{AM11} a model for product adoption
  in social networks with multiple products, where the agents,
  influenced by their neighbours, can adopt one out of several
  alternatives.  We identify and analyze here four types of paradoxes
  that can arise in these networks.  To this end, we use social
  network games that we recently introduced in \cite{SA12}.  These
  paradoxes shed light on possible inefficiencies arising when one
  modifies the sets of products available to the agents forming a
  social network.  One of the paradoxes corresponds to the well-known
  Braess paradox in congestion games and shows that by adding more
  choices to a node, the network may end up in a situation that is
  worse for everybody. We exhibit a dual version of this, where
  removing available choices from someone can eventually make
  everybody better off. The other paradoxes that we identify show that
  by adding or removing a product from the choice set of some node may
  lead to permanent instability.  Finally, we also identify conditions
  under which some of these paradoxes cannot arise.
\end{abstract}

\section{Introduction}
\label{sec:intro}


One of the most striking paradoxes in game theory is the Braess paradox.
It states that in some road networks the travel time can actually
increase when new roads are added, e.g., see \cite[pages 464-465]{NRTV07}.  
This paradox can be expressed as a statement about congestion
games, in which the players' objective is to minimize the travel time
from the source to the sink, and the cost (travel time) depends negatively
on the number of users of each road segment. Under this model, the Braess paradox
states that in some congestion games, an addition of a new strategy
(road segment) can trigger a sequence of changes (an improvement
path) that brings the players from the initial Nash equilibrium to a new one 
with worse travel time for each player.

The Braess paradox has received a great deal of attention because of its
counterintuitive character and potential implications. It has been studied
in other contexts too, for example queueing networks, see \cite{CK90}. A 
natural `dual' version of this paradox, concerning the removal
of road segments, has also been studied, see \cite{FKLS12,FKS12}.
This version states
that in some congestion games a removal of a strategy (road segment)
can trigger a sequence of changes (an improvement path) that brings
the players from the initial Nash equilibrium to a new one with better
travel time for each player.

Our main contribution is to demonstrate that paradoxes similar to the Braess
paradox and its dual version exist in a natural class of social
network models concerned with the diffusion of multiple products. These paradoxes provide us
with insights into the possible changes triggered by an addition or
removal of products in the considered social networks.  We also establish that
apart of the above two paradoxes two other types of paradoxes
exist.  In particular, it is possible that an addition of a new
product to (respectively, a removal of a product from) the choice set
of a player results in a permanent instability, in the sense that the
sequence of triggered changes may fail to terminate. Furthermore, 
we analyze variants of these paradoxes that are for example obtained by
stipulating that the corresponding `new situation' is inevitable
instead of only being possible.

Social networks have developed over the years into a large interdisciplinary research area with
important links to sociology, economics, epidemiology, computer
science, and mathematics. Regarding the diffusion of information over social networks, a flurry of numerous articles, notably the
influential \cite{Mor00}, and books, see
\cite{Cha04,Goy07,Veg07,Jac08,EK10}, helped to delineate better this
area. It deals with such diverse topics as epidemics, spread of
certain patterns of social behaviour, effects of advertising, and
emergence of `bubbles' in financial markets.  The model of social networks
that we consider here was introduced in \cite{AM11} and more fully in
\cite{AM13}. In these networks the agents (players), influenced by
their neighbours, can adopt one out of several alternatives.  An
example of such a network is a group of people who choose providers of
mobile phones by taking into account the choice of their friends.

To analyze the dynamics of such networks, we introduced in \cite{SA12},
and more fully in \cite{SA12b}, a natural class of \emph{social network
  games}.  In these strategic games, the payoff of each player weakly
increases when more players choose the same product (strategy) as him
- exactly the opposite of what happens in congestion games. In the same manner
that congestion games allow us to frame the Braess paradox, the games 
under consideration here allow us to formalize and analyze the abovementioned
paradoxes in social networks with multiple products.

The general setup in which we study these paradoxes makes it possible
to interpret them as phenomena that can take place in any community
the members of which make choices by taking into account the choices
of others. An example is a `bubble' in a financial market, where a
decision of a trader to switch to some new financial product triggers
a sequence of transactions, as a result of which all traders involved
become worse off.

Further, it was noticed in a number of empirical studies that an
abundance of choices may sometimes lead to wrong decisions.  To quote
from \cite[page 38]{Gig08}:
\begin{quote}                                                          
  \emph{The freedom-of-choice paradox}. The more options one has, the
  more possibilities for experiencing conflict arise, and the more
  difficult it becomes to compare the options. There is a point where
  more options, products, and choices hurt both seller and consumer.
\end{quote}                                                                                               
Both phenomena can be naturally explained in our framework.

Apart from exhibiting these paradoxes, it is also natural to try to identify classes of social networks in which the
introduced paradoxes cannot arise. The last part of our work (Sections \ref{sec:nosource} and \ref{sec:cycle}) is devoted to such
an analysis.

The paper is organized as follows. In the next section we introduce
the background material.  In Sections~\ref{sec:vul}, \ref{sec:fra},
\ref{sec:red}, and \ref{sec:uns}, using the social network games, we
define formally and analyze four types of paradoxes.  Then, in
Section~\ref{sec:nosource} we consider the case of networks where the
underlying graph has no source nodes and provide sufficient conditions
ensuring that one of the main paradoxes cannot arise. Subsequently, we utilize this result in 
Section~\ref{sec:cycle}, in which we
study the special case where the underlying graph is a simple cycle.
Finally, in Section~\ref{sec:conc} we discuss future research
directions.

\section{Preliminaries}
\label{sec:prelim}

\subsection{Strategic games}

A \bfe{strategic game} for $n > 1$ players, written as $(S_1, \ldots, S_n,
p_1, \ldots, p_n)$, consists of a non-empty set $S_i$ of
\bfe{strategies} and a \bfe{payoff function} $p_i : S_1 \times \cdots
\times S_n \myra \mathbb{R}$,
for each player $i$.

Fix a strategic game
$
G := (S_1, \ldots, S_n, p_1, \ldots, p_n).
$
We denote $S_1 \times \cdots \times S_n$ by $S$, 
call each element $s \in S$
a \bfe{joint strategy},
denote the $i$th element of $s$ by $s_i$, and abbreviate the sequence
$(s_{j})_{j \neq i}$ to $s_{-i}$. Occasionally we write $(s_i,
s_{-i})$ instead of $s$.  

We call a strategy $s_i$ of player $i$ a \bfe{best response} to a
joint strategy $s_{-i}$ of his opponents if $ \fa s'_i \in S_i
\  p_i(s_i, s_{-i}) \geq p_i(s'_i, s_{-i})$. We call a joint strategy
$s$ a \bfe{Nash equilibrium} if each $s_i$ is a best response to
$s_{-i}$.
Further, we call a strategy $s_i'$ of player $i$ a \bfe{better
  response} given a joint strategy $s$ if $p_i(s'_i, s_{-i}) >
p_i(s_i, s_{-i})$.



By a \bfe{profitable deviation} we mean a pair $(s,s')$ of joint
strategies such that $s' = (s'_i, s_{-i})$ for some $s'_i$ and
$p_i(s') > p_i(s)$.  Further, when $s'_i$ is a best response to
$s_{-i}$, we call it a \bfe{best response deviation}. Following
\cite{MS96}, an \bfe{improvement path} (respectively, a \bfe{best
  response improvement path}) is a maximal sequence of profitable
deviations (respectively, best response deviations). Clearly, if a
(best response) improvement path is finite, then its last element is a
Nash equilibrium.

Given two joint strategies $s$ and $s'$ we write

\begin{itemize}
\item $s >_w s'$ if for all $i$, $p_i(s) \geq p_i(s')$ and for some $i$, $p_i(s) > p_i(s')$,

\item $s >_s s'$ if for all $i$, $p_i(s) > p_i(s')$.

\end{itemize}
When $s >_w s'$ (respectively, $s >_s s'$) holds we say that $s'$ is
\bfe{weakly worse} (respectively, \bfe{strictly worse}) than $s$.


\subsection{Social networks}

We are interested in strategic games defined over a specific
type of social networks recently introduced in \cite{AM11}, which we recall first.

Let $V=\{1,\ldots,n\}$ be a finite set of \bfe{agents} and $\sgraph=(V,E,\weight)$ 
a weighted directed graph with $\weight_{ij} \in [0,1]$ being the
weight of the edge $(i,j)$. 
Given a node $i$ of $G$, we denote by
$\neighbour(i)$ the set of nodes from which there is an incoming edge to $i$.
We call each $j \in \neighbour(i)$ a \oldbfe{neighbour} of $i$ in $G$.
We assume that for each node $i$ such that $\neighbour(i) \neq \ES$, $\sum_{j
\in \neighbour(i)} w_{ji} \leq 1$.
An agent $i \in V$ is said to be a
\bfe{source node} in $\sgraph$ if $\neighbour(i)=\emptyset$.

By a \bfe{social network} (from now on, just \bfe{network}) we mean a
tuple $\snet=(\sgraph,\products,\prodset,\theta)$, where 
\begin{itemize}
\item $G$ is a weighted directed graph, 

\item $\products$ is a finite set of alternatives or \bfe{products},

\item $\prodset$ is a function that 
assigns to each agent $i$ a non-empty set of products $\prodset(i)$
from which it can make a choice, 

\item $\theta$ is a \bfe{threshold
  function} that for each $i \in V$ and $t \in \prodset(i)$ yields a
value $\theta(i,t) \in (0,1]$.
\end{itemize}

Given such a network $\snet$, we denote by $\srcnodes(\snet)$ the set of
source nodes in the underlying graph $\sgraph$.  

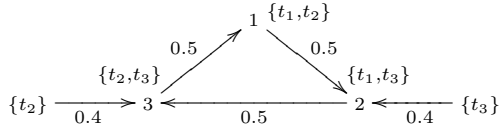
\begin{figure}[ht]
\centering
$
\def\objectstyle{\scriptstyle}
\def\labelstyle{\scriptstyle}
\xymatrix@R=20pt @C=30pt{
& &1 \ar[rd]^{0.5} \ar@{}[rd]^<{\{t_1,t_2\}}\\
\{t_2\} \ar[r]_{0.4} &3 \ar[ur]^{0.5} \ar@{}[ur]^<{\{t_2,t_3\}}& &2 \ar[ll]^{0.5} \ar@{}[lu]_<{\{t_1,t_3\}} &\{t_3\} \ar[l]^{0.4}\\
}$

\caption{\label{fig:socnet}A social network}
\end{figure}

\begin{example}
\label{ex:socnet}
\normalfont Figure \ref{fig:socnet} shows an example of a network. Let
the threshold be $0.3$ for all nodes. The set of products $\products$
is $\{t_1,t_2,t_3\}$, the product set of each agent is marked next to
the node denoting it and the weights are labels on the edges. Each
source node is represented by the unique product in its product set.
\HB
\end{example}

Given two social networks $\snet$ and $\snet'$ we say that $\snet'$ is
an \bfe{expansion} of $\snet$ if it results from adding a product to
the product set of a node in $\snet$.  We say then also that $\snet$
is a \bfe{contraction} of $\snet'$.

\subsection{Social network games}
\label{subsec:sng}

Next, we recall the strategic games over the social networks in the above sense that
we introduced in \cite{SA12}.
Fix a network $\snet=(\sgraph,\products,\prodset,\theta)$.
Each agent can adopt a product from his product set or
choose not to adopt any product. We denote the latter choice by
$t_0$. 

With each network $\snet$ we associate a strategic game
$\mathcal{G}(\snet)$. The idea is that the agents
simultaneously choose a product or abstain from choosing any.
Subsequently, each node assesses his choice by comparing it with the
choices made by his neighbours.  Formally, we define the game as
follows:

\begin{itemize}
\item the players are the agents (i.e., the nodes),

\item the set of strategies for player $i$ is
$S_i :=\prodset(i) \cup \{t_0\}$,

\item For $i \in V$, $t \in
\prodset(i)$ and a joint strategy $\strprofile$, let
$
\inflset_i^t(\strprofile) :=\{j \in \neighbour(i) \mid s_j=t\},
$
i.e., $\inflset_i^t(\strprofile)$ is the set of neighbours of $i$ who adopted in $s$
the product $t$.

The payoff function is defined as follows, where $\constutil$ is some given in advance
positive constant:

\begin{itemize}
\item for $i \in \srcnodes(\snet)$,

$\payoff_i(\strprofile) :=\left \{\begin{array}{ll}
                         0          & \mbox{if~~} \strprofile_i = t_0\\ 
                         \constutil & \mbox{if~~} \strprofile_i \in \prodset(i)\\
  \end{array}
  \right. $\\

\item for $i \not\in \srcnodes(\snet)$,

$\payoff_i(s) :=\left\{\begin{array}{ll}
               0 &\mbox{if~~} \strprofile_i = t_0\\
               \sum\limits_{j \in \inflset_i^t(\strprofile)} w_{ji}-\theta(i,t) & \mbox{if~~} \strprofile_i=t, \mbox{ for some } t \in \prodset(i)\\
\end{array}     
                           \right. $\\

\end{itemize}

\end{itemize}

In the first entry we assume that the payoff function for the source
nodes is constant only for simplicity.  The second entry in the payoff
definition is motivated by the following considerations.  When agent
$i$ is not a source node, his `satisfaction' from a joint strategy
depends positively from the accumulated weight (read: `influence') of
his neighbours who made the same choice as him, and negatively from
his threshold level (read: `resistance') to adopt this product.  The
assumption that $\theta(i,t) > 0$ reflects the view that there is
always some resistance to adopt a product. 
Strategy $t_0$ represents the possibility that an agent refrains from
choosing a product.  

\begin{example}
\label{ex:payoff}
\normalfont Consider the network given in Example~\ref{ex:socnet} and
the joint strategy $\strprofile$ where each source node chooses the unique
product in its product set and nodes 1, 2 and 3 choose $t_2$, $t_3$
and $t_2$ respectively. The payoffs are then given as follows:

\begin{itemize}
\item for the source nodes, the payoff is the fixed constant $\constutil$,
\item $\payoff_1(\strprofile)=0.5-0.3=0.2$,
\item $\payoff_2(\strprofile)=0.4-0.3=0.1$,
\item $\payoff_3(\strprofile)=0.4-0.3=0.1$.
\end{itemize}

Let $\strprofile'$ be the joint strategy in which player 3 chooses $t_3$
and the remaining players make the same choice as given in
$\strprofile$. Then $(\strprofile,\strprofile')$ is a profitable
deviation since $\payoff_3(\strprofile') > \payoff_3(\strprofile)$. In
what follows, we represent each profitable deviation by a node and a strategy it
switches to, e.g., $3:t_3$. Starting at $\strprofile$, the sequence of
profitable deviations $3:t_3, 1:t_0$ is an improvement path which
results in the joint strategy in which nodes 1, 2 and 3 choose $t_0$,
$t_3$ and $t_3$ respectively and each source node chooses the unique
product in its product set.  \HB
\end{example}

By definition, the payoff of each player depends only on the strategies
chosen by his neighbours, so the social network games are related to
graphical games of \cite{KLS01}. However, the underlying dependence
structure of a social network game is a directed graph and the
presence of the special strategy $t_0$ available to each player makes
these games more specific. Finally, note that these games satisfy the \bfe{join the
  crowd} property that we define as follows:

\begin{quote}
Each payoff function $p_i$ depends only on the strategy chosen by player $i$ and
the set of players who also chose his strategy. Moreover, 
the dependence on this set is monotonic.
\end{quote}

The last qualification is exactly opposite to the definition of 
congestion games with player-specific payoff functions of~\cite{Mil96}, 
in which the dependence on the above set is antimonotonic. 
That is, when more players choose the strategy of player $i$, then his payoff weakly decreases.






\section{Vulnerable networks}
\label{sec:vul}

In what follows we introduce and analyze four types of deficient networks.
In this section we focus on the following notions.

We say that a social network $\snet$ is $\te w$-\bfe{vulnerable} if
for some Nash equilibrium $s$ in $\mathcal{G}(\snet)$, an expansion
$\snet'$ of $\snet$ exists such that some improvement path in
$\mathcal{G}(\snet')$ leads from $s$ to a Nash equilibrium $s'$ in
$\mathcal{G}(\snet')$ such that $s >_w s'$.  In general we have four
notions of vulnerability, that correspond to the combinations $X Y$,
where $X \in \{\te, \fa\}$ and $Y \in \{w,s\}$. For example, we say
that $\snet$ is $\fa s$-\bfe{vulnerable} if for some Nash equilibrium
$s$ in $\mathcal{G}(\snet)$, an expansion $\snet'$ of $\snet$ exists
such that each improvement path in $\mathcal{G}(\snet')$ leads from
$s$ to a Nash equilibrium $s'$ in $\mathcal{G}(\snet)$ such that $s
>_s s'$.

First note that there are some obvious implications between the four
notions of vulnerability and inefficiency that we exhibit in
Figure~\ref{fig:vul}.

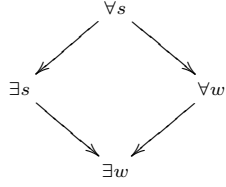
\begin{figure}[ht]
\centering
$
\def\objectstyle{\scriptstyle}
\def\labelstyle{\scriptstyle}
\xymatrix@W=10pt @R=20pt @C=20pt{
& \forall s \ar[ld]^{} \ar[rd]^{}\\
\exists s \ar[rd]^{}& & \forall w \ar[ld]^{}\\
& \exists w\\
}$
\caption{\label{fig:vul}Dependencies between the notions of vulnerability}
\end{figure}

We show now that these implications are the
only ones that hold between these four notions.



\begin{example}[$\fa w$] \label{exa:faw}
\label{ex:awvulnerable}
\rm
In Figure~\ref{fig:faw} we exhibit an example of a $\fa w$-vulnerable
network that is not $\te s$-vulnerable. The product set of each node
is marked next to it and the weights are labels on the edges. We
assume that each threshold is a constant $\theta$, where $0 < \theta <
0.1$.  Here and elsewhere the relevant expansion is depicted by means
of a product and the dotted arrow pointing to the relevant node. In this case
product $t_1$ is added to node 4.

The initial Nash equilibrium $s$ is the joint strategy formed by the
underlined products, i.e., $(t_2, t_3, t_3, t_3, t_1, t_1)$.  Consider
now what happens after product $t_1$ is added to the product set of
node 4.  Then $s$ ceases to be a Nash equilibrium.  Addition of $t_1$
triggers the unique best response improvement path
\[
4:t_1, 3:t_2, 5:t_2, 6:t_0, 4:t_3, 3:t_3, 5:t_0 
\]
resulting in the Nash equilibrium $(t_2,t_3,t_3,t_3,t_0,t_0)$. Note
that at each step of any improvement path starting in $\strprofile$
triggered by the addition of product $t_1$ to node 4 there is a unique
node which is not playing its best response. For instance, in the
second step of the above improvement path, node 3 is the unique node
which is not playing its best response. Although node 3 can profitably
deviate to $t_0$ instead of $t_2$, in the next step which is unique,
node 3 is forced to play its best response. Therefore it suffices to
consider the outcome of the above best response improvement path. In
the Nash equilibrium $(t_2,t_3,t_3,t_3,t_0,t_0)$ the payoffs of
players 1--4 did not change with respect to the original Nash
equilibrium, while the payoffs of players 5 and 6 decreased.

Finally, notice that a network is
not $\te s$-vulnerable if the underlying graph has a source node.
\HB










\begin{figure}[ht]
\centering
$
\def\objectstyle{\scriptstyle}
\def\labelstyle{\scriptstyle}
\xymatrix@W=10pt @R=40pt @C=30pt{
1 \ar[d]_{0.1} \ar@{}[d]_<{\{\underline{t_2}\}}&  &2 \ar[d]^{0.1} \ar@{}[d]^<{\{\underline{t_3}\}}\\
3 \ar@{}[d]_<{\{t_2,\underline{t_3}\}} \ar[d]_{0.3} \ar@/^0.7pc/[rr]^{0.2}& &4  \ar@/^0.7pc/[ll]^{0.2} \ar@{}[d]^<{\{\underline{t_3}\}}& t_1 \ar@{..>}[l] \\
5 \ar@{}[u]^<{\{\underline{t_1},t_2\}} \ar@/^0.7pc/[rr]^{0.2}& &6 \ar@/^0.7pc/[ll]^{0.2} \ar@{}[u]_<{\{\underline{t_1}\}} \ar[u]_{0.3} \\
}$
\caption{\label{fig:faw}An $\fa w$-vulnerable network}
\end{figure}
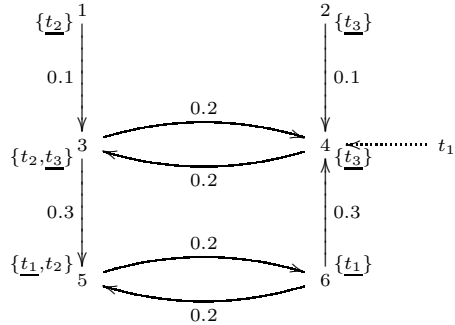

\end{example}

In this specific example the payoff of the player who triggered the
change in the end did not change.  A slightly more complicated
example, that we omit, shows that the
initiator's payoff in the final Nash equilibrium can decrease. Also, one
can construct examples in which the payoffs in the final Nash
equilibrium decrease for an arbitrary large fraction of the players
and remain constant for the other players.

\begin{example}[$\te s$]

\rm
In Figure~\ref{fig:tes} we exhibit an example of a $\te s$-vulnerable
network that is not $\fa w$-vulnerable (and hence not $\fa s$-vulnerable).  
As before we assume that each threshold is a constant
$\theta$, where $0 < \theta < 0.1$ and we underline the strategies
that form the initial Nash equilibrium.

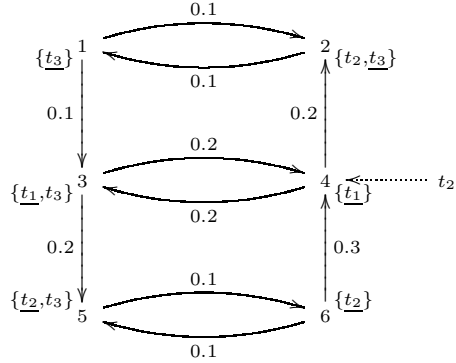
\begin{figure}[ht]
\centering
$
\def\objectstyle{\scriptstyle}
\def\labelstyle{\scriptstyle}
\xymatrix@W=10pt @R=40pt @C=30pt{
1 \ar[d]_{0.1} \ar@{}[d]_<{\{\underline{t_3}\}} \ar@/^0.7pc/[rr]^{0.1}&  &2 \ar@/^0.7pc/[ll]^{0.1} \ar@{}[d]^<{\{t_2,\underline{t_3}\}}\\
3 \ar@{}[d]_<{\{\underline{t_1},t_3\}} \ar[d]_{0.2} \ar@/^0.7pc/[rr]^{0.2}& &4 \ar[u]^{0.2} \ar@/^0.7pc/[ll]^{0.2} \ar@{}[d]^<{\{\underline{t_1}\}}& t_2 \ar@{..>}[l] \\
5 \ar@{}[u]^<{\{\underline{t_2},t_3\}} \ar@/^0.7pc/[rr]^{0.1}& &6 \ar@/^0.7pc/[ll]^{0.1} \ar@{}[u]_<{\{\underline{t_2}\}} \ar[u]_{0.3} \\
}$
\caption{\label{fig:tes}A $\te s$-vulnerable network}
\end{figure}

To see that this network is $\te s$-vulnerable it suffices to note
that starting from the Nash equilibrium $(t_3, t_3, t_1, t_1, t_2,
t_2)$ of the initial network
the addition of product $t_2$ to node 4 triggers the best response improvement path
\[
4:t_2, 3:t_3, 5:t_3, 6:t_0, 2:t_2, 1:t_0, 4:t_0, 2:t_0, 3:t_0, 5:t_0
\]
that ends in a Nash equilibrium in which each strategy equals $t_0$,
and consequently
each payoff becomes 0.

To see that this network is not $\fa w$-vulnerable first note that
the addition of product $t_2$ to node 4 also triggers the improvement path
\[
4:t_2, 2:t_2, 1:t_0, 3:t_0
\]
that ends in a Nash equilibrium in which the payoffs of nodes 2 and 4 increase.

Each of the remaining initial Nash equilibria includes a strategy
$t_0$. But when a joint
strategy includes $t_0$, then for no $s'$ we have $s >_s s'$. This
allows us to conclude that
the considered network is not $\fa s$-vulnerable.
To show that the considered network is in fact not $\fa w$-vulnerable
we need to analyze
each of the initial Nash equilibria and consider all profitable additions.
We only consider one representative example. Consider
the initial Nash equilibrium $(t_0, t_0, t_1, t_1, t_2, t_2)$ and the
addition of product
$t_1$ to node 2. This triggers the unique improvement path $2:t_1,
1:t_0$ that ends in
a Nash equilibrium in which the payoff of node 2 increased.
\HB
\end{example}

If we just wish to construct an improvement path, so not necessarily a
best response improvement path, that yields a strictly worse Nash
equilibrium, then a simpler example can be used. Namely, one can drop
in the above network the nodes 1 and 2 and all the arcs to and from
them, and adjust the threshold function of node 3 so that $\theta(3,
t_3) < \theta(3, t_1)$.
Then 
\[
4:t_2, 3:t_3, 5:t_3, 6:t_0, 4:t_0, 3:t_0, 5:t_0
\]
is the desired improvement path.

\begin{example}[$\te w$]\label{exa:tew}
  
\rm

Next, we provide an example of a $\te w$-vulnerable network that is
neither $\te s$-vulnerable nor $\fa w$-vulnerable.  It suffices to add
to the network given in Figure \ref{fig:tes} a source node $7$ with
the product set $\{t_1\}$ and connect it to node $1$ using an
arbitrary threshold and weight. In each Nash equilibrium, node $7$
chooses $t_1$, so its payoff is the same. Further, the choice of this
node has no influence on the choices of other nodes in the Nash
equilibria in the original and the extended networks.  So the
conclusion follows from the previous example.  
\HB
\end{example}

Next, we would like to mention the following intriguing question:
\II

\NI
\textbf{Open problem:}
Do $\fa s$-vulnerable networks exist?
\II

The following result
shows that if they do, they use at least three products.

\begin{theorem}
When there are only two products $\fa w$-vulnerable networks, so a
fortiori $\fa s$-vulnerable networks, do not exist.
\end{theorem}
\begin{proof}

Suppose by contradiction that such a network $\snet$ exists. So a Nash equilibrium $s$
in $\mathcal{G}(\snet)$, a node, say 1, and a product, say $t_1$,
exists such that for the network expansion $\snet'$ obtained by adding
$t_1$ to the product set of node 1 each improvement path that starts
in $s$ ends up in a Nash equilibrium $s'$ in $\mathcal{G}(\snet')$
such that $s >_w s'$.  

Given an initial joint strategy we call a maximal sequence of best
response deviations to a strategy $t$ (in an arbitrary order) a
\bfe{$t$-phase}.  We now repeatedly perform, starting at $s$, the
$t_1$-phase followed by the $t_0$-phase. We claim that this process
terminates and hence yields a finite improvement path in
$\mathcal{G}(\snet')$.

First note that if a joint strategy $s^2$ is obtained from $s^1$ by having
some nodes to switch to product $t_1$ and $t_1$ is a best response for
a node $i$ to $s^1_{-i}$, then $t_1$ is also a best response for $i$ to
$s^2_{-i}$. Indeed, by the join the crowd property 
$p_{i}(t_1, s^2_{-i}) \geq p_{i}(t_1, s^1_{-i})$ and 
$p_{i}(t_2, s^1_{-i}) \geq p_{i}(t_2, s^2_{-i})$, so 
$p_{i}(t_1, s^2_{-i}) \geq p_{i}(t_2, s^2_{-i})$ since 
$p_{i}(t_1, s^1_{-i}) \geq p_{i}(t_2, s^1_{-i})$. Further, 
$p_{i}(t_1, s^1_{-i}) \geq p_{i}(t_0, s^1_{-i})$, so also 
$p_{i}(t_1, s^2_{-i}) \geq p_{i}(t_0, s^2_{-i})$.
Consequently after the first $t_1$-phase each node that has the
strategy $t_1$ plays a best response. Call the outcome of the
first $t_1$-phase $s''$.

Consider now a node $i$ that deviated in $s''$ to $t_0$ by means of a
best response. By the observation just made node $i$ deviated from
product $t_2$. So, again by the join the crowd property, this
deviation does not affect the property that the nodes that selected
$t_1$ in $s''$ play a best response.  Iterating this reasoning we
conclude that after the first $t_0$-phase each node that has the
strategy $t_1$ continues to play a best response.

By the same reasoning subsequent $t_1$ and
$t_0$-phases have the same effect on the set of nodes that have the strategy
$t_1$: each of these nodes continues to play a best response. 

Moreover, this set continues to weakly increase.  Consequently these
repeated applications of the $t_1$-phase followed by the $t_0$-phase
terminate, say in a joint strategy $s'$.  Suppose now a node $i$ does
not play a best response to $s'_{-i}$.  If $s'_i = t_0$, then by the
construction $t_1$ is not a best response, so $t_2$ is a best
response.

Suppose the initial strategy of node $i$ was also $t_0$, i.e., $s_i =
t_0$.  Since $s$ is a Nash equilibrium in $\mathcal{G}(\snet)$, we
have $p_i(t_2, s_{-i}) \leq p_i(t_0, s_{-i})$.  By the join the crowd
property $p_i(t_2, s'_{-i}) \leq p_i(t_2, s_{-i})$, so $p_i(t_2,
s'_{-i}) \leq p_i(t_0, s'_{-i})$, which yields a contradiction.  Hence
node $i$ deviated to $t_0$ from some intermediate joint strategy $s^1$
by selecting a best response.  So $p_i(t_2, s^1_{-i}) \leq p_i(t_0,
s^1_{-i})$. Moreover, by the join the crowd property $p_i(t_2,
s'_{-i}) \leq p_i(t_2, s^1_{-i})$, so $p_i(t_2, s'_{-i}) \leq p_i(t_0,
s'_{-i})$, which yields a contradiction, as well.

Further, by the construction $s'_i \neq t_1$, so the only alternative is that
$s'_i = t_2$. But then either $t_0$ or $t_1$ is a best response, which
contradicts the construction of $s'$.
We conclude that $s'$ is a Nash equilibrium in $\mathcal{G}(\snet')$.

Next, the payoff of node 1 strictly increased when it switched to
$t_1$ and, on the account of the above arguments, during the remaining steps
of the considered improvement path it either increased or remained the
same.  We conclude that the final Nash equilibrium $s'$ is not weakly
worse than the original, which yields a contradiction.
\end{proof}

\section{Fragile networks}
\label{sec:fra}

Related notions to vulnerable networks are the following ones.

We say that a social network $\snet$ is $\te$-\bfe{fragile} if for
some Nash equilibrium $s$ in $\mathcal{G}(\snet)$, an expansion
$\snet'$ of $\snet$ exists such that some improvement path in
$\mathcal{G}(\snet')$ that starts in $s$ is infinite.  In turn, we say
that a social network $\snet$ is $\fa$-\bfe{fragile} if for some Nash
equilibrium $s$ in $\mathcal{G}(\snet)$, an expansion $\snet'$ of
$\snet$ exists such that each improvement path in
$\mathcal{G}(\snet')$ that starts in $s$ is infinite.  Finally, we say
that a social network $\snet$ is \bfe{fragile} if $\mathcal{G}(\snet)$
has a Nash equilibrium, while for some expansion $\snet'$ of $\snet$,
$\mathcal{G}(\snet')$ does not.

Obviously each fragile network is $\fa$-fragile, while each $\fa$-fragile network
is $\te$-fragile. We now show that these two implications are proper.

\begin{example}[Fragile]
\label{ex:fragile}
\rm 
Consider the network $\snet$ given in Figure \ref{fig:fragile} where
the source nodes are represented by the unique product in their
product set. We assume that each threshold is a constant $\theta$ such
that $\theta < w_1 < w_2$.

\begin{figure}[ht]
\centering
$
\def\objectstyle{\scriptstyle}
\def\labelstyle{\scriptstyle}
\xymatrix@R=22pt @C=30pt{
& & \{t_1\} \ar[d]^{w_1}\\
& &1 \ar[rd]^{w_2} \ar@{}[u]^<{\{t_1\}} & t_2 \ar@{..>}[l]\\
\{t_2\} \ar[r]_{w_1} &3 \ar[ur]^{w_2} \ar@{}[ur]^<{\{t_2,t_3\}}& &2 \ar[ll]^{w_2} \ar@{}[lu]_<{\{t_1,t_3\}} &\{t_3\} \ar[l]^{w_1}\\
}$

\caption{\label{fig:fragile}A fragile network}
\end{figure}
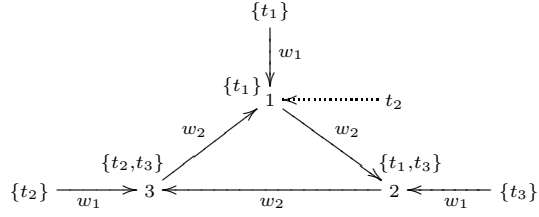
Consider the joint strategy $\strprofile$, in which the nodes marked
by $\{t_1\}$, $\{t_2\}$ and $\{t_3\}$ choose the unique product in
their product set and nodes 1, 2, and 3 choose $t_1, t_1$ and $t_2$,
respectively. For convenience, we denote $\strprofile$ by the choices
of nodes 1, 2 and 3, so $\strprofile=(t_1,t_1,t_2)$.  It is easy to
verify that $\strprofile$ is a Nash equilibrium in
$\mathcal{G}(\snet)$. 

Consider now the expansion $\snet'$ of $\snet$ in
which product $t_2$ is added to the product set of node 1. In
$\mathcal{G}(\snet')$ the joint strategy $s$ ceases to remain a Nash
equilibrium. In fact, no joint strategy is a Nash equilibrium in
$\mathcal{G}(\snet')$. Each agent residing on the triangle can secure a payoff of
at least $w_1 - \theta>0$, so it suffices to analyze the joint
strategies in which $t_0$ is not used. There are in total eight such
joint strategies. Here is their listing, where in each joint strategy
we underline the strategy that is not a best response to the choice of
other players: $(\underline{t_1}, t_1, t_2)$, $(t_1, t_1,
\underline{t_3})$, $(t_1, t_3, \underline{t_2})$, $(t_1,
\underline{t_3}, t_3)$, $(t_2, \underline{t_1}, t_2)$, $(t_2,
\underline{t_1}, t_3)$, $(t_2, t_3, \underline{t_2})$,
$(\underline{t_2}, t_3, t_3)$. This shows that the initial network
$\snet$ is fragile.  
\HB
\end{example}

\begin{example}[$\forall$-fragile]
\label{ex:afragile}
\rm
Consider the network $\snet$ given in Figure \ref{fig:afragile}.
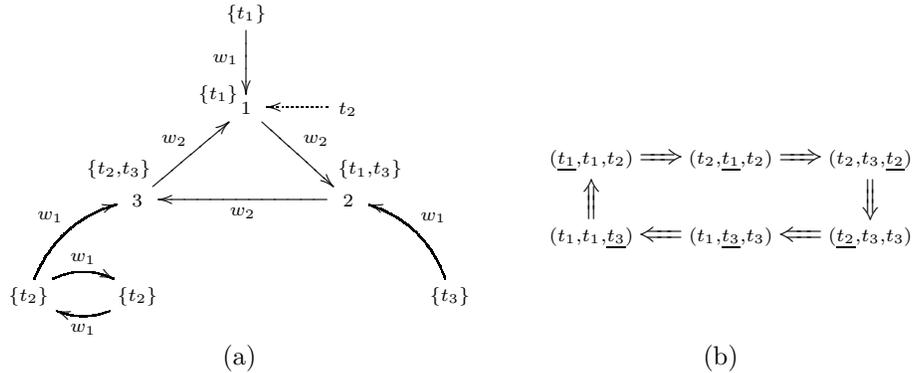
\begin{figure}[ht]
\centering
\begin{tabular}{ccc}
$
\def\objectstyle{\scriptstyle}
\def\labelstyle{\scriptstyle}
\xymatrix@W=10pt @C=20pt{
& & \{t_1\} \ar[d]_{w_1} & \\
& &1 \ar[rd]^{w_2} \ar@{}[u]^<{\{t_1\}} & t_2 \ar@{..>}[l]\\
 &3 \ar[ur]^{w_2} \ar@{}[ur]^<{\{t_2,t_3\}}& &2 \ar[ll]^{w_2} \ar@{}[lu]_<{\{t_1,t_3\}} \\
\{t_2\} \ar@/^0.7pc/[ru]^{w_1} \ar@/^0.7pc/[r]^{w_1} & \{t_2\} \ar@/^0.7pc/[l]^{w_1}& & &\{t_3\} \ar@/_0.7pc/[lu]_{w_1}
}$
&
&
\raisebox{-2.5cm}{\parbox{8cm}{
$
\def\objectstyle{\scriptstyle}
\def\labelstyle{\scriptstyle}
\xymatrix@W=10pt @C=15pt @R=15pt{
(\underline{t_1},t_1,t_2)\ar@{=>}[r]& (t_2,\underline{t_1},t_2)\ar@{=>}[r]& (t_2,t_3,\underline{t_2})\ar@{=>}[d]\\
(t_1,t_1,\underline{t_3})\ar@{=>}[u]& (t_1,\underline{t_3},t_3)\ar@{=>}[l]& (\underline{t_2},t_3,t_3)\ar@{=>}[l]\\
}$
}}
\\
(a) & & \parbox{3.7cm}{(b)}\\
\end{tabular}
\caption{\label{fig:afragile}A $\forall$-fragile network and an
  infinite improvement path}
\end{figure}
We assume that each threshold is a constant $\theta$, where $\theta <
w_1 < w_2$. Consider the joint strategy $\strprofile$, in which the
nodes marked by $\{t_1\}$, $\{t_2\}$ and $\{t_3\}$ choose the unique
product in their product set and nodes 1, 2, and 3 choose $t_1, t_1$
and $t_2$, respectively. As in the previous example we denote
$\strprofile$ by $(t_1,t_1,t_2)$. It is easy to check that
$\strprofile$ is a Nash equilibrium in $\mathcal{G}(\snet)$. 

Now consider the expansion $\snet'$ obtained by adding the product
$t_2$ to the product set of node 1. The joint strategy $\strprofile$
ceases to remain a Nash equilibrium in $\mathcal{G}(\snet')$. In fact,
Figure \ref{fig:afragile}(b) shows the unique best response
improvement path starting in $\strprofile$ which is infinite. For each
joint strategy in the figure, we underline the strategy that is not a
best response. As in the case of Example \ref{ex:awvulnerable} at
every step of every improvement path starting in $\strprofile$, there
is a unique node which is not playing its best response. Therefore it
suffices to consider the above best response improvement path. This
shows that $\snet$ is $\forall$-fragile. 

Also note that the game $\mathcal{G}(\snet')$ has a Nash equilibrium.
The joint strategy in which the nodes marked $\{t_2\}$ along with node
3 choose the product $t_0$, the node marked $\{t_3\}$ choose the
product $t_3$ and the node marked $\{t_1\}$ along with nodes 1 and 2
choose $t_1$ forms a Nash equilibrium. It follows that $\snet$ is not
fragile.
\HB
\end{example}

\begin{example}[$\exists$-fragile]
\label{ex:efragile}
\rm
Consider the network $\snet$ given in Figure \ref{fig:efragile}(a). 
\begin{figure}[ht]
\centering
\begin{tabular}{ccc}
$
\def\objectstyle{\scriptstyle}
\def\labelstyle{\scriptstyle}
\xymatrix@W=10pt @C=20pt{
& \{t_4\} \ar[dr]_{w_3} &\{t_1\} \ar[d]^{w_1}\\
& &1 \ar[rd]^{w_2} \ar@{}[rd]^<{\{t_1,t_2,t_4\}}\\
 &3 \ar[ur]^{w_2} \ar@{}[ur]^<{\{t_2,t_3,t_4\}}& &2 \ar[ll]^{w_2} \ar@{}[lu]_<{\{t_3,t_4\}} &\\
\{t_2\} \ar@/^0.7pc/[ru]^{w_1}& & & t_1 \ar@{..>}[u] &\{t_3\} \ar@/_0.7pc/[lu]_{w_1}\\
}$
&
&
\raisebox{-2.5cm}{\parbox{8cm}{
$
\def\objectstyle{\scriptstyle}
\def\labelstyle{\scriptstyle}
\xymatrix@W=10pt @C=15pt @R=15pt{
(t_1,t_1,\underline{t_3})\ar@{=>}[r]& (\underline{t_1}, t_1, t_2)\ar@{=>}[r]& (t_2,\underline{t_1},t_2)\ar@{=>}[d]\\
(t_1,\underline{t_3},t_3)\ar@{=>}[u]& (\underline{t_2},t_3,t_3)\ar@{=>}[l]& (t_2,t_3,\underline{t_2})\ar@{=>}[l]\\
}$
}}
\\
(a) & & \parbox{3.7cm}{(b)}\\
\end{tabular}
\caption{\label{fig:efragile} A $\exists$-fragile network and an
  infinite improvement path}
\end{figure}
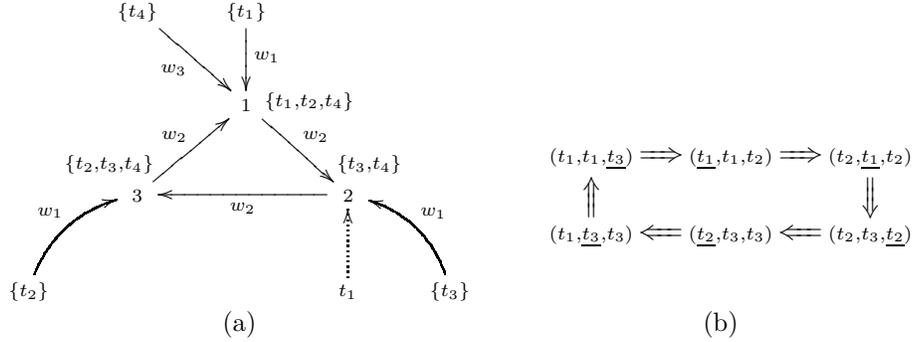
Let the threshold be a constant $\theta$, where $\theta < w_3 < w_1 <
w_2$. Assume that each source node selects its unique
product. Identify each joint strategy that extends this selection with
the selection of the strategies by the nodes 1, 2 and 3. The joint
strategy $\strprofile=(t_1,t_3,t_3)$ is a Nash equilibrium in
$\mathcal{G}(\snet)$. 

Now consider the expansion $\snet'$ obtained by
adding the product $t_1$ to the product set of node 2 in $\snet$. The
joint strategy $\strprofile$ ceases to remain a Nash equilibrium in
$\mathcal{G}(\snet')$ since node 2 can profitably deviate to
$t_1$. Figure \ref{fig:efragile}(b) shows an infinite improvement path
starting in $(t_1,t_1,t_3)$. Therefore $\snet$ is $\exists$-fragile.

However, $\snet$ is not $\forall$-fragile. First, one can check that
$\strprofile=(t_1,t_3,t_3)$ is the only Nash equilibrium in
$\mathcal{G}(\snet)$ for which profitable additions exist. Below we
analyse the two profitable additions.

\begin{itemize}
\item Addition of $t_1$ to node 2. The following improvement path
\[2:t_1, 3:t_2, 1:t_4, 2:t_4, 3:t_4 \]
starting in $\strprofile$ terminates in the joint strategy
$(t_4,t_4,t_4)$ which is a Nash equilibrium.
\item Addition of $t_3$ to node 1. This triggers a unique one-step
  improvement path that terminates in a new Nash equilibrium
  $(t_3,t_3,t_3)$.
\HB
\end{itemize}
\end{example}

\section{Inefficient networks}
\label{sec:red}

The last two types of deficiency are concerned with product
removal. These form the dual versions of the paradoxes we have seen so far. 
In this section we study the following notions.

We say that a social network $\snet$ is $\te w$-\bfe{inefficient} if
for some Nash equilibrium $s$ in $\mathcal{G}(\snet)$, a contraction
$\snet'$ of $\snet$ exists such that some improvement path in
$\mathcal{G}(\snet')$ leads from $s$ to a Nash equilibrium $s'$ in
$\mathcal{G}(\snet')$ such that $s' >_w s$.  We note here that if the
contraction was created by removing a product from the product set of
node $i$, we impose that any improvement path in
$\mathcal{G}(\snet')$, given a starting joint strategy from
$\mathcal{G}(\snet)$, begins by having node $i$ making a choice (we
allow any choice from his remaining set of products as an improvement
move). Otherwise the initial payoff of node $i$ in
$\mathcal{G}(\snet')$ is not well-defined.

As in the case of the vulnerability, we have
four notions of inefficiency that correspond to the combinations $X
Y$, where $X \in \{\te, \fa\}$ and $Y \in \{w,s\}$. For example, we
say that $\snet$ is $\fa s$-\bfe{inefficient} if for some
Nash equilibrium $s$ in $\mathcal{G}(\snet)$, a contraction $\snet'$
of $\snet$ exists such that each improvement path in
$\mathcal{G}(\snet')$ leads from $s$ to a Nash equilibrium $s'$ in
$\mathcal{G}(\snet')$ such that $s' >_s s$.

We now show that the implications between the various notions shown in
Figure~\ref{fig:vul} for the case of vulnerable networks are also
proper implications for inefficient networks.  However, in contrast to
the concept of vulnerability, here even when there are only two
choices ($|P|=2$), there exist $\fa s$-inefficient networks.

\begin{example}[$\fa s$] \label{exa:forall-s-contraction}
\rm

We exhibit in Figure \ref{fig:forall-s-contraction} an example of a $\fa
s$-inefficient network. The weight of each edge is assumed to be $w$, and 
we also have the same
product-independent threshold, $\theta$, for all nodes, with $w> \theta$.



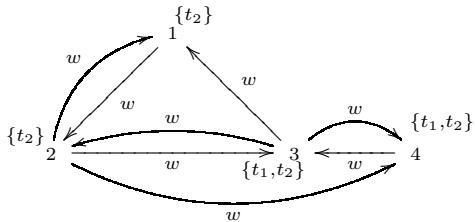
\begin{figure}[ht]
\centering
$
\def\objectstyle{\scriptstyle}
\def\labelstyle{\scriptstyle}
\xymatrix@W=10pt @R=35pt @C=30pt{
& 1 \ar[dl]^{w} \ar@{}[rr]^<{\{t_2\}}& & \\
2 \ar[rr]_{w} \ar@/^1pc/[ru]^{w} \ar@/_1.6pc/[rrr]_{w} \ar@{}@/^1pc/[ru]^<{\{t_2\}}& & 3 \ar[lu]_{w} \ar@{}[ll]^<{\{t_1,t_2\}} \ar@/^1pc/[r]^{w} \ar@/_0.7pc/[ll]_{w} & 4 \ar[l]^{w} \ar@{}[lu]_<{\{t_1,t_2\}}\\
}$
\caption{\label{fig:forall-s-contraction}An example of a $\forall s$-inefficient network}
\end{figure}

Consider as the initial Nash equilibrium the joint strategy $s = (t_2, t_2,
t_1, t_1)$. It is easy to check that this is a Nash equilibrium,
with a payoff equal to $w - \theta$ for all nodes. 
Suppose now that we remove
product $t_1$ from the product set of node $3$. We claim that all
improvement paths then lead to the Nash equilibrium in which all nodes adopt
$t_2$. 

To see this, we simply analyze all the cases that can arise.  Since
node $3$ moves first in the improvement path, it can select either
$t_0$ or $t_2$. If it selects $t_2$, then it will never change again
his decision because it has a positive payoff due to node $1$, and
$t_0$ will never become a better choice. Hence it will then be the
turn of node $4$ to switch. If it switches to $t_0$,
then  it will have to change again and switch to $t_2$
because of the support for the choice of $t_2$ (from node
$1$). This means that eventually all nodes switch
to $t_2$.

Suppose now that in the beginning node $3$ selected $t_0$.
Subsequently, node $4$ can also switch to $t_0$.  However the main
observation here is that $t_0$ is not a best response for any of these
nodes. Both nodes receive a positive payoff for adopting $t_2$ thanks
to a support from node $1$. Hence no matter which of them switches to
$t_0$ at the beginning of the improvement path, eventually they will
both switch to $t_2$.

To conclude, all improvement paths result in all nodes adopting $t_2$
and producing a payoff of $2w - \theta$ for each node, which is
strictly better than the payoff in $s$.  
\HB
\end{example}

\begin{example}[$\fa w$] \label{exa:forall-w-contraction}

\rm
We now exhibit a network which is $\fa w$-inefficient but not $\te
s$-inefficient (and hence also not $\fa s$-inefficient). We proceed as
in Example~\ref{exa:tew} and add to the network given in Figure
\ref{fig:forall-s-contraction} a source node $5$ with the product set
$\{t_1\}$ and connect it to node $1$, using the same weight $w$ and
threshold $\theta$.  By the same argument as in Example~\ref{exa:tew}
the conclusion follows by virtue of the previous example.
\HB
\end{example}

We also remark that one can construct even simpler networks with three
nodes and two products that exhibit the same behaviour.
%

\begin{example}[$\te s$] \label{exa:exists-s-contraction}

\rm
Next, we exhibit a network that is $\te s$-inefficient but not $\fa
w$-inefficient. The network is shown in Figure
\ref{fig:exists-s-contraction}.  The weight of each edge is assumed to
be $w$, and we also have a product-independent threshold $\theta$
(with $w> \theta$), that applies to all nodes and products except
three cases: $\theta(1, t_3) < \theta$, and $\theta(5, t_2) =
\theta(5, t_3) > \theta$. Note that in the underlying graph, each node
has exactly two incoming edges, one from the set $\{1, 2\}$ and one
from $\{3, 4, 5\}$.


\begin{figure}[ht]
\centering
$
\def\objectstyle{\scriptstyle}
\def\labelstyle{\scriptstyle}
\xymatrix@W=10pt @R=35pt @C=30pt{
& 2 \ar[dl]^{w} \ar[rr]^{w} \ar@{}[rr]^<{\{t_2,t_3\}}& &4 \ar[dr]^{w} \ar@{}[rd]^<{\{t_1,t_2,t_3\}}\\
1 \ar[rr]_{w} \ar@/^1pc/[ru]^{w} \ar@/_1.6pc/[rrrr]_{w} \ar@{}@/^1pc/[ru]^<{\{t_2\}}& & 3 \ar[lu]_{w} \ar@{}[ll]^<{\{t_1,t_2,t_3\}} \ar[ru]_{w} \ar@/_0.7pc/[ll]_{w} & & 5 \ar[ll]_{w} \ar@{}[lu]_<{\{t_1,t_2,t_3\}}\\
}$
\caption{\label{fig:exists-s-contraction}An example of a $\exists s$-inefficient network}
\end{figure}
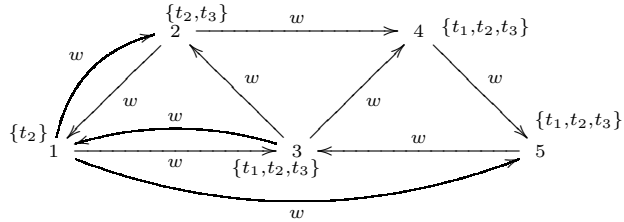

To see first that this is a $\te s$-inefficient network, consider the
Nash equilibrium $(t_2, t_2, t_1, t_1, t_1)$. In this joint strategy
each node receives suport from exactly one out of its two
neighbours. If we delete $t_1$ from the choice set of node $3$, then
we can see that there is an improvement path that converges to all
nodes adopting $t_2$ (by having first node $3$ adopt $t_2$, followed
by nodes $4$ and $5$). Hence in this new Nash equilibrium every node
receives support from all its neighbours and the payoff of everyone
has strictly increased.

To argue now that this is not a $\fa w$-inefficient network, we need
to consider all Nash equilibria and argue about all the possible
contractions.  One can verify that the initial game has four Nash
equilibria, namely $(t_2, t_2, t_1, t_1, t_1)$, $(t_2, t_2, t_2, t_2,
t_2)$, $(t_0, t_0, t_1, t_1, t_1)$, and $(t_0, t_3, t_3, t_3,
t_3)$. The idea behind this example is that in any contraction, either
node $1$ or node $2$ or node $5$ will get worse off in some
improvement path. This will happen either because nodes $1$ or $2$
will lose support for $t_2$ or because node $5$ will end up with
product $t_3$, which is a worse choice for him than $t_1$.  

Let us analyze the first Nash equilibrium. Note that if the
contraction deletes $t_2$ from the product set of node $1$ or $2$,
then node $2$ will end up being worse off. If the contraction deletes
$t_1$ from node $3$, then consider the following improvement path
\[3:t_3, 4:t_3, 5:t_3, 1:t_3.\]
The last profitable deviation in this path is ensured by our
assumption that $\theta(1, t_3) < \theta$. This implies that node $2$
is worse off at the end. Similar improvement paths can be found if the
contraction involves nodes $4$ or $5$.

If we consider the second Nash equilibrium, $(t_2, t_2, t_2, t_2,
t_2)$, it is even easier to see that any contraction makes at least
one node worse off in some improvement path since all nodes apart from
node $5$ receive the maximum possible payoff under this joint strategy.

Next, we analyze the third Nash equilibrium, $(t_0, t_0, t_1, t_1,
t_1)$. Here the most interesting contraction is to remove $t_1$ from
node $3$ (the same things hold if we remove it from nodes $4$ or
$5$). In that case, we can have the improvement path
\[3:t_3, 4:t_3, 5:t_3, 1:t_3  \]
 where node $5$ becomes worse off since $\theta(5, t_3) > \theta$. The
 analysis for the fourth Nash equilibrium is similar and omitted.
\HB
\end{example}

\begin{example}[$\te w$] \label{exa:exists-w-contraction}

\rm
Finally, we exhibit a $\te w$-inefficient network that is neither $\fa
w$-inefficient nor $\te s$-inefficient.  We proceed as in
Example~\ref{exa:forall-w-contraction} and simply modify the previous
example.  We add to the network given in Figure
\ref{fig:exists-s-contraction} a source node $6$ with the product set
$\{t_2\}$ and connect it to node $1$, using the same weight $w$ and
threshold $\theta$.  By the same argument as in Examples~\ref{exa:tew}
and~\ref{exa:forall-w-contraction} 
the conclusion follows by virtue of the previous example.

Actually, the argument that this network is not $\fa w$-inefficient
becomes now simpler because after the addition of the source node $6$
the initial game has only two Nash equilibria, namely $(t_2, t_2, t_1,
t_1, t_1, t_2)$ and $(t_2, t_2, t_2, t_2, t_2, t_2)$.
\HB
\end{example}

\section{Unsafe networks}
\label{sec:uns}

Finally, we have three notions that are counterparts of
the fragility notions. We say that a social network $\snet$ is
$\te$-\bfe{unsafe} (respectively, $\fa$-\bfe{unsafe}) if for some Nash
equilibrium $s$ in $\mathcal{G}(\snet)$, a contraction $\snet'$ of
$\snet$ exists such that some (respectively, each) improvement path in
$\mathcal{G}(\snet')$ that starts in $s$ is infinite.  Further, a
social network $\snet$ is \bfe{unsafe} if $\mathcal{G}(\snet)$ has a
Nash equilibrium, while for some contraction $\snet'$ of $\snet$,
$\mathcal{G}(\snet')$ does not.

Analogously to Section~\ref{sec:fra} each unsafe network is
$\fa$-unsafe, while each $\fa$-unsafe network is $\te$-unsafe. We now
prove that these two implications are proper.

\begin{example}[Unsafe] 
\label{ex:unsafe}
\normalfont Let $\snet_1$ be the modification of the network $\snet$
given in Figure~\ref{fig:fragile} where node 1 and the source node
marked with $\{t_1\}$ has the product set $\{t_1,t_2\}$. Consider the
joint strategy in which this source node along with node 1 choose
$t_2$, nodes 2 and 3 choose $t_3$ and nodes marked by $\{t_2\}$ and
$\{t_3\}$ choose the unique product in their product set. This is a
Nash equilibrium in $\snet_1$. Now consider the contraction $\snet_2$
of $\snet_1$ in which the product $t_2$ is removed from the source
node with product set $\{t_1,t_2\}$. Then $\snet_2$ is same as the
network $\snet'$ in Example~\ref{ex:fragile}. Following the argument
in Example~\ref{ex:fragile} we conclude that the initial network
$\snet_1$ is unsafe.
\HB
\end{example}

\begin{example}[$\forall$-unsafe]
\label{ex:aunsafe}
\normalfont Let $\snet_1$ be the modification of the network $\snet$
given in Figure~\ref{fig:afragile} where node 1 and the source node
marked $\{t_1\}$ has the product set $\{t_1,t_2\}$. Consider the joint
strategy in which this source node along with node 1 choose $t_2$,
nodes 2, 3 and the node marked $\{t_3\}$ choose $t_3$ and nodes marked
by $\{t_2\}$ choose the unique product in their product set. This is a
Nash equilibrium in $\snet_1$. Now consider the contraction $\snet_2$
of $\snet_1$ in which the product $t_2$ is removed from the source
node with product set $\{t_1,t_2\}$. Then $\snet_2$ is same as the network
$\snet'$ in Example~\ref{ex:afragile}. Following the argument in
Example~\ref{ex:afragile} we conclude that the initial network
$\snet_1$ is $\forall$-unsafe but not unsafe.
\HB
\end{example}

\begin{example}[$\exists$-unsafe]
\label{ex:eunsafe}
\normalfont Let $\snet_1$ be the modification of the network $\snet$
given in Figure~\ref{fig:efragile} where node 2 has the product set
$\{t_1,t_3,t_4\}$ and the source node marked with $\{t_3\}$ has the
product set $\{t_1,t_3\}$. Consider the joint strategy in which this
source node along with node 2 choose $t_1$, nodes 1 and 3 choose $t_2$
and nodes marked by $\{t_1\}$, $\{t_2\}$ and $\{t_4\}$ choose the
unique product in their product set. This is a Nash equilibrium in
$\snet_1$. Now consider the contraction $\snet_2$ of $\snet_1$ in
which the product $t_1$ is removed from the source node with product
set $\{t_1,t_3\}$. Then $\snet_2$ is same as the network $\snet'$ in
Example~\ref{ex:efragile}. Following the argument in
Example~\ref{ex:efragile} we conclude that the initial network
$\snet_1$ is $\exists$-unsafe but not $\forall$-unsafe.
\HB
\end{example}

\section{Networks without source nodes}
\label{sec:nosource}

Given the variety of paradoxes exhibited in the above examples it is
natural to investigate the status of selected networks. In this
section we focus first on networks where there are no source nodes.
This is a reasonable assumption in social networks as everybody
usually has some friends who influence his decisions. 
We first identify a property which ensures the non-existence of
$\exists w$-vulnerable networks, when the underlying graph has no
source nodes. 

For a joint strategy $\strprofile$ and product $t$, let
$\agents_t(\strprofile):=\{i \in V \mid \strprofile_i=t\}$, and
$\mathit{prod}(\strprofile):=\{\strprofile_i \mid i \in V\} \setminus
\{t_0\}$. Hence, $\mathit{prod}(\strprofile)$ is the set of distinct strategies that are used in profile $\strprofile$.
We let also $\obar{t}$ denote the joint strategy in which every
player selects $t$. We say that a profile $\strprofile$ is a \bfe{multiple product} profile, if
$|\mathit{prod}(\strprofile)|>1$.

\begin{theorem}
\label{thm:not-exists-w}
Consider a network $\snet$ whose underlying graph has no source nodes.
If $\mathcal{G}(\snet)$ does not have a multiple product Nash
equilibrium, then $\snet$ is not $\exists w$-vulnerable.
\end{theorem}

To prove this result, we use a specific structural property of Nash
equilibria in networks whose underlying graph has no source nodes.
Below, we only consider subgraphs that are \emph{induced} and identify
each such subgraph with its set of nodes. Recall that $(V',E')$ is an
induced subgraph of $(V,E)$ if $V' \sse V$ and $E' = E \cap (V' \times
V')$. For subgraphs $C_1$ and $C_2$, we denote by $C_1 \cap C_2$ the
intersection of the nodes of the graphs. We say that a (non-empty)
strongly connected subgraph (in short, SCS) $C$ of $\sgraph$ is
\bfe{self sustaining} for a product $t$ if for all $i \in C$,

\begin{itemize}
\item $t \in \prodset(i)$,
\item $\sum\limits_{j \in \neighbour(i)\cap C} w_{ji} \geq
  \theta(i,t)$.
\end{itemize}

Hence, $C$ is a self sustaining SCS for a product $t$ if assigning this
product to every node in $C$ ensures that each node in $C$ gets a
non-negative payoff. A self sustaining SCS $C$ is \bfe{minimal} for
product $t$ if no subgraph $C'$ of $C$ is a self sustaining SCS
for product $t$. First we prove the following auxiliary result.

\begin{lemma}
\label{lm:Ne-SCSstruct}
Let $\snet=(\sgraph,\products,\prodset,\theta)$ be a network whose
underlying graph has no source nodes. If $\strprofile \neq \obar{t_0}$
is a Nash equilibrium in $\mathcal{G}(\snet)$, then for all products
$t \in \mathit{prod}(\strprofile)$, there exists a minimal self
sustaining SCS $C$ for $t$ such that $C \subseteq
\agents_t(\strprofile)$.
\end{lemma}

\begin{proof}
Suppose $s \neq \obar{t_0}$ is a Nash equilibrium. Take any product $t
\neq t_0$ and an agent $i$ such that $\strprofile_i=t$ (by assumption,
at least one such $t$ and $i$ exists).  Recall that
$\inflset_i^t(\strprofile)$ denotes the set of neighbours of $i$ who
adopted in $s$ the product $t$.  Consider the set of nodes
$\spred:=\bigcup_{m \in \nat} \spred_m$, where
\begin{itemize}
\item $\spred_0:=\{i\}$,

\item $\spred_{m+1}:=\spred_m \bigcup \big( \bigcup_{j \in \spred_m} \inflset_j^t(\strprofile) \big)$.

\end{itemize}

By construction, for all $j \in \spred$ we have $s_j = t$ and
$\inflset_j^t(\strprofile) \sse \spred$. Since
$\strprofile$ is a Nash equilibrium, also $\sum\limits_{k \in
  \inflset_j^t(\strprofile)} w_{kj} \geq \theta(j,t)$ holds.

Consider the partial ordering $<$ between the strongly connected components of
the graph induced by $\spred$
defined by: $C < C'$ iff $j \to k$ for some $j \in C$ and $k \in C'$.
Take now some SCS $C'$ induced by a strongly connected component that
is minimal
in the $<$ ordering. 
Then for all $k \in C'$ we have $\inflset_k^t(\strprofile) \sse
C'$ and hence $\inflset_k^t(\strprofile) \sse \neighbour(k)\cap
C'$.  This shows that $C'$ is self sustaining. It is then
straightforward to construct a minimal self sustaining SCS $C$ for
product $t$ which is a subgraph of $C'$.
\end{proof}

Given a network $\snet$ and a product $t \in \products$, let
$\mathcal{C}_t(\snet)$ be the set of all minimal self sustaining SCSs
for product $t$. Let $X_t(\snet)=\bigcap_{C \in \mathcal{C}_t(\snet),
  \mathcal{C}_t(\snet) \neq \emptyset} C$ and $Y(\snet)=\bigcap_{t\in
  \products, X_t(\snet) \neq \emptyset}X_t(\snet)$.


\medskip

\noindent{\it Proof of Theorem \ref{thm:not-exists-w}:} To prove the
theorem, we show in fact the following claim. Suppose that for the
network $\snet=(\sgraph,\products,\prodset,\theta)$ one of the following conditions
holds: 
\begin{enumerate} 
\item for all $t \in \products$, $\mathcal{C}_t(\snet) = \emptyset$,
\item $Y(\snet) \neq \emptyset$.  
\end{enumerate}
Then, $\snet$ is not $\exists w$-vulnerable.

In other words, if the network $\snet$ does not have any self
sustaining SCSs or if the intersection of the set of all minimal self
sustaining SCSs is non-empty then $\snet$ is not $\exists
w$-vulnerable.  Note that condition $1$ implies that $\obar{t_0}$ is
the only Nash equilibrium in $\mathcal{G}(\snet)$ and condition $2$
implies that $|\mathit{prod}(\strprofile)|=1$ for any Nash equilibrium
$\strprofile$ in $\mathcal{G}(\snet)$ that is different from
$\obar{t_0}$.

First suppose that for all $t \in \products$, $\mathcal{C}_t(\snet) =
\emptyset$. By Lemma~\ref{lm:Ne-SCSstruct} it follows that
$\obar{t_0}$ is the only Nash equilibrium in $\snet$. Consider any
expansion $\snet'$ of $\snet$. Then no player has a profitable
deviation from $\obar{t_0}$ in $\snet'$. Therefore, $\snet$ is not
$\exists w$-vulnerable.

Now suppose that $Y(\snet) \neq \emptyset$. In this case, we first
claim that every non-trivial Nash equilibrium $\strprofile$ in $\snet$
has the property that $\mathit{prod}(\strprofile) \subseteq \{t_1\}$
for some $t_1 \in \products$. Suppose this is not the case then for
some two different products $t_1$ and $t_2$, $\{t_1,t_2\} \subseteq
\mathit{prod}(\strprofile)$. Since $\strprofile$ is a Nash
equilibrium, by Lemma~\ref{lm:Ne-SCSstruct}, there exists a minimal
self sustaining SCS $C_1 \subseteq \agents_{t_1}(\strprofile)$ for
$t_1$ and a minimal self sustaining SCS $C_2 \subseteq
\agents_{t_2}(\strprofile)$ for $t_2$. By definition,
$\agents_{t_1}(\strprofile) \cap \agents_{t_2}(\strprofile)
=\emptyset$ and therefore, $\mathcal{C}_{t_1}(\snet) \cap
\mathcal{C}_{t_2}(\snet) = \emptyset$. This contradicts the assumption
that $Y(\snet) \neq \emptyset$.

Consider a Nash equilibrium $\strprofile$ in $\snet$ and an expansion
$\snet'$. By the above claim $\mathit{prod}(\strprofile)\subseteq
\{t_1\}$ for some $t_1 \in \products$.  In the expansion $\snet'$, if
the new product $t_2$ is added to a node $i$, where $\strprofile_i =
t_1$, then there is no profitable deviation from $\strprofile$ and
consequently, no improvement path starting at $\strprofile$ in
$\mathcal{G}(\snet')$. Thus the interesting case is when the new
product $t_2$ is added to a node $i$ where $\strprofile_i=t_0$. We
have two cases.
\begin{itemize}
\item \emph{Case 1:} $t_2 \neq t_1$. Since $\mathit{prod}(\strprofile)
  \subseteq \{t_1\}$, there is no profitable deviation from
  $\strprofile$ and therefore, no improvement path starting at
  $\strprofile$ in $\mathcal{G}(\snet')$.
\item \emph{Case 2:} $t_2 = t_1$. Consider any improvement path in
  $\mathcal{G}(\snet')$ that leads to a Nash equilibrium
  $\strprofile'$ in $\mathcal{G}(\snet')$. Since $\strprofile$ is a
  Nash equilibrium in $\mathcal{G}(\snet)$ the first profitable
  deviation in the improvement path is of the form $i:t_1$. In the
  improvement path, if a joint strategy $s^2$ is obtained from $s^1$
  by having some nodes switch to product $t_1$ and $t_1$ is a best
  response for a node $j$ to $s^1_{-j}$, then $t_1$ is also a best
  response for $j$ to $s^2_{-j}$. Indeed, by the join the crowd
  property $p_{j}(t_1, s^2_{-j}) \geq p_{j}(t_1, s^1_{-j}) \geq
  p_j(t_1,s^1_{-j}) = 0 = p_j(t_0,s^2_{-j})$. So the only deviations
  in this improvement path are to $t_1$.
Consequently in $\strprofile'$ which is a Nash equilibrium, $t_1$ is
the product selected by node $i$ (i.e., $\strprofile_i'=t_1$) and
$\payoff_i(\strprofile') > \payoff_i(\strprofile)$. Therefore, the
network is not $\exists w$-vulnerable.
\end{itemize}
\HB

\section{Simple cycle networks}
\label{sec:cycle}

In this section we focus on networks where the underlying graph is a
simple cycle. To fix the notation suppose that the underlying graph is
$1 \to 2 \to \LL \to n \to 1$.  We assume that the counting is done in
cyclic order within \C{\LLn} using the increment operation $i \oplus
1$ and the decrement operation $i \ominus 1$. In particular, $n \oplus
1 = 1$ and $1 \ominus 1 = n$. We start with the following corollary to
Theorem \ref{thm:not-exists-w}.

We begin with the following observation.

\begin{note} \label{not:simple}
Consider a simple cycle network. Then $s$ is a Nash equilibrium
of the game $\mathcal{G}(\snet)$ iff $s = \bar{t}$
for some $t \in \products \cup \{t_0\}$.
\end{note}
\begin{proof}
Consider a network
$\snet=(\sgraph,\products,\prodset,\theta)$, where $\sgraph$ is a
simple cycle.  
First note that $\obar{t_0}$ is always a Nash equilibrium. 
Consider now a Nash
equilibrium $\strprofile \neq \obar{t_0}$. Then there exists a product
$t$ and a node $i$ such that $s_i = t$. Since $\strprofile$ is a Nash
equilibrium, we have $p_i(s) \geq p_i(t_0, s_{-i}) = 0$, so $s_{i
  \ominus 1} = t$ as well (otherwise the node $i$ would have negative payoff). Iterating this reasoning we conclude that
$s = \obar{t}$. 
\end{proof}

\begin{corollary}
\label{cor:scycle-vul}
Simple cycle networks are not $\exists w$-vulnerable (and a fortiori
not $X Y$-vulnerable, where $X \in \{\te, \fa\}$ and $Y \in \{w,s\}$).
\end{corollary}
\begin{proof}
This is an immediate consequence of Note~\ref{not:simple}
and Theorem~\ref{thm:not-exists-w}.
\end{proof}

The remaining types of deficiency are easy to determine.
For the case of fragile networks we prove the following result.

\begin{theorem}
Simple cycle networks are not $\exists w$-fragile
(and a fortiori not $\forall w$-fragile and not fragile).
\end{theorem}

\begin{proof}
  Consider a simple cycle network
  $\snet=(\sgraph,\products,\prodset,\theta)$, a Nash equilibrium $s$
  of $\mathcal{G}(\snet)$, and an expansion $\snet'$ of $\snet$.  By
  Note~\ref{not:simple} $s = \obar{t}$, where $t \in
  \products \cup \{t_0\}$. Hence $s$ remains a Nash equilibrium of
  $\mathcal{G}(\snet')$.  Consequently $\snet$ is not $\exists
  w$-fragile.
\end{proof}

In the case of inefficient networks we have the following result.
\begin{theorem}
\mbox{}
\begin{enumerate}[(i)]
\item There exists a simple cycle network $\snet$ that is $\exists
  s$-inefficient (and a fortiori $\exists w$-inefficient).

\item Simple cycle networks are not $\forall w$-inefficient (and a fortiori not $\forall s$-inefficient).
\end{enumerate}
\end{theorem}

\begin{proof}
$(i)$ Consider the network shown in Figure
  \ref{fig:inef_cycle}. Suppose that $\theta(i, t_1) > \theta(i, t_2)$
  for all nodes $i= 1,2,3$ and that $s = \obar{t_1}$ is a Nash
  equilibrium. Starting from $s$, suppose we remove $t_1$ from the
  product set of node $1$.  Then there exists a finite improvement
  path where all the nodes end up adopting $t_2$, by simply having
  node $1$ adopt $t_2$ and then having the remaining nodes follow
  their best response. Since $\theta(i, t_1) > \theta(i, t_2)$, all
  nodes are strictly better off in this new Nash equilibrium,
  $\obar{t_2}$.

\begin{figure}[ht]
\centering
$
\def\objectstyle{\scriptstyle}
\def\labelstyle{\scriptstyle}
\xymatrix@W=10pt @C=15pt{
 &1 \ar[rd]^{} \ar@{}[rd]^<{\{t_1,t_2\}} \ar@{}[rd]^>{\{t_1,t_2\}}\\
3  \ar[ru]_{} \ar@{}[ru]^<{\{t_1,t_2\}} & &2 \ar[ll]^{}\\
}$
\caption{\label{fig:inef_cycle} An $\exists s$-inefficient simple cycle network}
\end{figure}
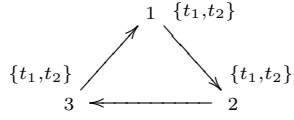

\NI 
$(ii)$ Consider a simple cycle network
$\snet=(\sgraph,\products,\prodset,\theta)$, a Nash equilibrium $s$ of
$\mathcal{G}(\snet)$, and a contraction $\snet'$ of $\snet$.  By
Note~\ref{not:simple} $s = \obar{t}$, where $t \in
\products \cup \{t_0\}$. If $s = \obar{t_0}$, then it is impossible
that by deleting any product we could make some player better off.
Suppose $s = \obar{t_1}$ for some product $t_1$. If we delete $t_1$
from some product set, say of node $1$, then there is always an
improvement path that terminates at the Nash equilibrium $\obar{t_0}$
(simply start with node $1$ adopting $t_0$, and proceed clockwise.
Then gradually every other node will switch to $t_0$ since they
eventually lose support for $t_1$). Hence no node is better off in
this new Nash equilibrium. In conclusion, there can be no Nash
equilibrium from which all improvement paths after the contraction
will make the set of nodes weakly better off.
\end{proof}

Finally, we consider the case of unsafe networks.

\begin{theorem}
\mbox{}
\begin{enumerate}[(i)]
\item There exists a simple cycle network $\snet$ that is $\exists$-unsafe.

\item Simple cycle networks are not $\forall$-unsafe (and a fortiori
  not unsafe).
\end{enumerate}
\end{theorem}

\begin{proof}
$(i)$ 
Consider the network shown in Figure \ref{fig:cycle}(a)
and assume that $\obar{t}$ is a Nash equilibrium.

\begin{figure}[ht]
\centering
\begin{tabular}{ccc}
$
\def\objectstyle{\scriptstyle}
\def\labelstyle{\scriptstyle}
\xymatrix@W=10pt @C=15pt{
 &1 \ar[rd]^{} \ar@{}[rd]^<{\{t\}} \ar@{}[rd]^>{\{t\}}\\
3  \ar[ru]_{} \ar@{}[ru]^<{\{t,t_1\}} & &2 \ar[ll]^{}\\
}$
&
&
$
\def\objectstyle{\scriptstyle}
\def\labelstyle{\scriptstyle}
\xymatrix@W=8pt @C=15pt @R=15pt{
(\underline{t},t, t_0)\ar@{=>}[r]& (t_0, t, \underline{t_0})\ar@{=>}[r]& (t_0, \underline{t}, t)\ar@{=>}[d]\\
(t, \underline{t_0}, t_0)\ar@{=>}[u]& (t, t_0, \underline{t})\ar@{=>}[l]& (\underline{t_0}, t_0, t)\ar@{=>}[l]\\
}$
\\
\\
(a) & & (b)\\
\end{tabular}
\caption{\label{fig:cycle}A simple cycle network and an infinite improvement path}
\end{figure}
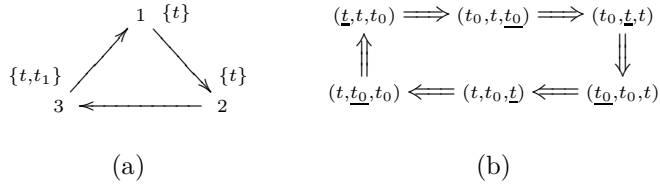 

By removing from the product set of node 3 the product $t$ we get in the resulting game 
an infinite improvement path depicted in Figure~\ref{fig:cycle}(b).
(In each joint strategy we underline the strategy that
is not a best response to the choice of other players.)
So the initial network is $\exists$-unsafe.
\II

\NI
$(ii)$
By Theorem 28 in~\cite{SA12b} for every simple cycle network $\snet$
there exists a finite improvement path in $\mathcal{G}(\snet)$. This implies
both claims.
\end{proof}


The above analysis does not carry through to all strongly
connected graphs. Indeed, we showed in particular that simple
cycle networks cannot be $\fa w$-vulnerable and also not $\forall
s$-inefficient. However, in Example~\ref{exa:faw} we exhibited a
network that is $\fa w$-vulnerable and in
Example~\ref{exa:forall-s-contraction} a network that is $\forall
s$-inefficient.  The underlying graphs of both networks are
strongly connected.

\section{Conclusions}
\label{sec:conc}

In this paper we provided a systematic study of paradoxes that can
arise in social networks with multiple products. Such paradoxes allow us to better understand possible undesirable
consequences of modifying the choices that are available to the agents
within a social network. The focus of our work was on identifying
these paradoxes and on determining their relative strength.

To analyze them, we
used a natural game-theoretic framework in the form of social network
games recently introduced in~\cite{SA12} and ~\cite{SA12b}. These
games do not always admit (pure) Nash equilibria, and as a result, more
types of paradoxes can arise (as we exhibited), than in the class of congestion games with
its celebrated Braess paradox. Out of all the notions of paradoxes that we introduced and studied, 
one question still remained open: do $\fa s$-vulnerable networks exist?

In future work we plan to assess the computational complexity of
determining the presence of these paradoxes and plan to analyze other selected
networks.  We also plan to expand our analysis of selected classes of networks,
determining which paradoxes can then be present.

Finally, in our analysis we assumed that the agents can refrain from
selecting a product.  Recently we initiated in \cite{AS13} a study of
an alternative version of social network games, in which each agent
has to choose a product.  This corresponds to natural situations, for
instance when pupils have to choose a primary school or when each
student has to select a laptop.  Such social networks are studied by
modifying our framework so that the strategy $t_0$ is not available.
This change leads to a different analysis and different results.

\bibliographystyle{abbrv}

\bibliography{/ufs/apt/bib/e}

\end{document}